\documentclass[11pt,a4paper]{article}
\usepackage{vmargin}
\setmarginsrb{1.0in}{1.0in}{1.0in}{1.0in}{0mm}{0mm}{5mm}{5mm}

\usepackage{amsfonts,amsmath,amssymb,stmaryrd}
\usepackage{amsthm}
\usepackage{graphicx}
\usepackage[utf8]{inputenc}
\usepackage[T1]{fontenc}

\usepackage{hyperref}

\newtheorem{lemma}{Lemma}
\newtheorem{theorem}[lemma]{Theorem}
\newtheorem{corollary}[lemma]{Corollary}

\newtheorem{observation}[lemma]{Observation}
\usepackage[numbers,sort&compress,longnamesfirst]{natbib}
\newtheorem{definition}[lemma]{Definition}

\newcommand{\hide}[1]{\relax}
\newcommand{\cost}{{\tt distcost}}
\newcommand{\completecost}{{\tt cost}}

\title{On the Tree Conjecture for the Network Creation Game}

\author{Davide Bilò\thanks{Department of Humanities and Social Sciences, University of Sassari, Italy, \texttt{davidebilo@uniss.it}} 
\and Pascal Lenzner\thanks{Algorithm Engineering Group, Hasso Plattner Institute Potsdam, Germany \texttt{pascal.lenzner@hpi.de}}
}

\date{~}
\begin{document}

\maketitle

\begin{abstract}
\noindent Selfish Network Creation focuses on modeling real world networks from a game-theoretic point of view. One of the classic models by Fabrikant et al.~[PODC'03] is the network creation game, where agents correspond to nodes in a network which buy incident edges for the price of $\alpha$ per edge to minimize their total distance to all other nodes. The model is well-studied but still has intriguing open problems. The most famous conjectures state that the price of anarchy is constant for all $\alpha$ and that for $\alpha \geq n$ all equilibrium networks are trees.

We introduce a novel technique for analyzing stable networks for high edge-price $\alpha$ and employ it to improve on the best known bounds for both conjectures. In particular we show that for $\alpha > 4n-13$ all equilibrium networks must be trees, which implies a constant price of anarchy for this range of $\alpha$. Moreover, we also improve the constant upper bound on the price of anarchy for equilibrium trees.
 \end{abstract}
 
 \section{Introduction}
Many important networks, e.g. the Internet or social networks, have been created in a decentralized way by selfishly acting agents. Modeling and understanding such networks is an important challenge for researchers in the fields of Computer Science, Network Science, Economics and Social Sciences. A significant part of this research focuses on assessing the impact of the agents' selfish behavior on the overall network quality measured by the \emph{price of anarchy}~\cite{KP99}. Clearly, if there is no or little coordination among the egoistic agents, then it cannot be expected that the obtained networks minimize the social cost. The reason for this is that each agent aims to improve the network quality for herself while minimizing the spent cost. However, empirical observations, e.g. the famous small-world phenomenon~\cite{TM67,K00,Bar16}, suggest that selfishly built networks are indeed very efficient in terms of the overall cost and of the individually perceived service quality. Thus, it is a main challenge to justify these observations analytically. 

A very promising approach towards this justification is to model the creation of a network as a strategic game which yields networks as equilibrium outcomes and then to investigate the quality of these networks. For this, a thorough understanding of the structural properties of such equilibrium networks is the key.  

We contribute to this endeavor by providing new insights into the structure of equilibrium networks for one of the classical models of selfish network creation~\cite{Fab03}. In this model, agents correspond to nodes in a network and can buy costly links to other nodes to minimize their total distance in the created network. Our insights yield improved bounds on the price of anarchy and significant progress towards settling the so-called tree conjecture~\cite{Fab03,MMM13}.

\subsection{Model and Definitions}
We consider the classical \emph{network creation game} as introduced by Fabrikant et al.~\cite{Fab03}. There are $n$ agents $V$, which correspond to nodes in a network, who want to create a connected network among themselves. 
Each agent selfishly strives for minimizing her cost for creating network links while maximizing her own connection quality. All edges in the network are undirected and unweighted and agents can create any incident edge for the price of $\alpha>0$, where $\alpha$ is a fixed parameter of the game. The strategy $S_u \subseteq V\setminus\{u\}$ of an agent $u$ denotes which edges are bought by this agent, that is, agent $u$ is willing to create (and pay for) all the edges $(u,x)$, for all $x \in S_u$. Let $\mathbf{s}$ be the $n$-dimensional vector of the strategies of all agents. The strategy-vector $\mathbf{s}$ induces an undirected network $G(\mathbf{s}) = (V,E(\mathbf{s}))$, where for each edge $(u,v) \in E(\mathbf{s})$ we have $v \in S_{u}$ or $u \in S_{v}$. If $v \in S_{u}$, then we say that agent $u$ is the \emph{owner} of edge $(u,v)$ or that agent $u$ \emph{buys} the edge $(u,v)$, otherwise, if $u \in S_{
v}$, then agent $v$ owns or buys the edge $(u,v)$.\footnote{No edge can have two owners in any equilibrium network. Hence, we will assume throughout the paper that each edge in $E(\mathbf{s})$ has a unique owner.} Since the created networks will heavily depend on $\alpha$ we emphasize this by writing $(G(\mathbf{s}),\alpha)$ instead of $G(\mathbf{s})$.
The \emph{cost} of an agent $u$ in the network $(G(\mathbf{s}),\alpha)$ is the sum of her cost for buying edges, called the \emph{creation cost}, and her cost for using the network, called the \emph{distance cost}, which depends on agent $u$'s distances to all other nodes within the network. The cost of $u$ is defined as $$\completecost(G(\mathbf{s}),\alpha,u) = \alpha|S_u|+\cost(G(\mathbf{s}),u),$$ where the distance cost is defined as 
$$\cost(G(\mathbf{s}),u) = \begin{cases}
                        \sum_{w\in V} d_{G(\mathbf{s})}(u,w), & \text{if $G(\mathbf{s})$ is connected}\\
                        \infty, & \text{otherwise.}
                                                                                                                                                                      \end{cases}$$
                                                                                                                                                                      Here $d_{G(\mathbf{s})}(u,w)$ denotes the length of a shortest path between $u$ and $w$ in the network $G(\mathbf{s})$. We will mostly omit the reference to the strategy vector, since it is clear that a strategy vector directly induces a network and vice versa. Moreover, if the network is clear from the context, then we will also omit the reference to the network, e.g. writing $\cost(u)$ instead of $\cost(G,u)$. 
                                                                                                                                                                      
A network $(G(\mathbf{s}),\alpha)$ is in \emph{pure Nash equilibrium} (NE), if no agent can unilaterally change her strategy to strictly decrease her cost. That is, in a NE network no agent can profit by a strategy change if all other agents stick to their strategies. Since in a NE network no agent wants to change the network, we call them \emph{stable}.

The \emph{social cost}, denoted $\completecost(G(\mathbf{s}),\alpha)$, of a network $(G(\mathbf{s}),\alpha)$ is the sum of the cost of all agents, that is, $\completecost(G(\mathbf{s}),\alpha) = \sum_{u\in V}\completecost(G(\mathbf{s}),\alpha,u)$. Let OPT$_n$ be the minimum social cost of a $n$ agent network and let maxNE$_n$ be the maximum social cost of any NE network on $n$ agents. The \emph{price of anarchy} (PoA)~\cite{KP99} is the maximum over all $n$ of the ratio $\frac{\text{maxNE}_n}{\text{OPT}_n}$.

Let $G = (V,E)$ be any undirected connected graph with $n$ vertices. 
A \emph{cut-vertex} $x$ of $G$ is a vertex with the property that $G$ with vertex $x$ removed contains at least two connected components. 
We say that $G$ is \emph{biconnected} if $n \geq 3$ and $G$ contains no cut-vertex. A \emph{biconnected component} $H$ of $G$ is a maximal induced subgraph of $G$ which is also biconnected. Note that we rule out trivial biconnected components which contain exactly one edge. Thus, there exist at least two vertex-disjoint paths between any pair of vertices $x,y$ in a biconnected component $H$, which implies that there exists a simple cycle containing $x$ and~$y$.

\subsection{Related Work}
Network creation games, as defined above, were introduced by Fabrikant et al.~\cite{Fab03}. They gave the first general bound of $\mathcal{O}(\sqrt{\alpha})$ on the PoA 
and they conjectured that above some constant edge-price all NE networks are trees. This conjecture, called the \emph{tree conjecture}, is especially interesting since they also showed that tree networks in NE have constant PoA. In particular, they proved that the PoA of stable tree networks is at most $5$. Interestingly, the tree conjecture in its general version was later disproved by Albers et al.~\cite{Al14}. However, non-tree NE networks are known only when $\alpha < n$, in particular, for every $\varepsilon > 0$, there exist non-tree NE networks with $\alpha \leq n-\varepsilon$~\cite{MMM13}. It is believed that for $\alpha \geq n$ the tree conjecture may be true. Settling this claim is currently a major open problem and there has been a series of papers which improved bounds concerning the tree conjecture.

First, Albers et al.~\cite{Al14} proved that for $\alpha \geq 12 n\log n$ every NE network is a tree. Then, using a technique based on the average degree of the biconnected component, this was significantly improved to $\alpha > 273n$ by Mihal\'{a}k \& Schlegel~\cite{MS10} and even further to $\alpha \geq 65n$ by Mamageishvili et al.~\cite{MMM13}. 
The main idea of this average degree technique is to prove a lower and an upper bound on the average degree of the unique biconnected component in any equilibrium network. The lower bound has the form "for $\alpha > c_1 n$ the average degree is at least $c_2$" and the upper bound has the form "for $\alpha > c_3 n$ the average degree is at most $f(\alpha)$", where $c_1, c_2, c_3$ are constants and $f$ is a function which monotonically decreases in $\alpha$. For large enough $\alpha$ both bounds contradict each other, which proves that equilibrium networks for this $\alpha$ cannot have a biconnected component and thus must be trees.
Very recently a preprint by Àlvarez \& Messegué~\cite{AM17} was announced which invokes the average degree technique with a stronger lower bound. This then yields a contradiction already for $\alpha > 17n$. For their stronger lower bound the authors use that in every minimal cycle (we call them "min cycles") of an equilibrium network all agents in the cycle buy exactly one edge of the cycle. This fact has been independently established by us~\cite{L14} and we also use it.  

The currently best general upper bound of $2^{\mathcal{O}(\sqrt{\log n})}$ on the PoA is due to 
Demaine et al.~\cite{De07} and it is known that the PoA is constant if $\alpha < n^{1-\epsilon}$ for any fixed $\epsilon \geq \frac{1}{\log n}$~\cite{De07}. Thus, the PoA was shown to be constant for almost all $\alpha$, except for the range between $n^{1-\epsilon}$, for any fixed $\epsilon \geq \frac{1}{\log n}$, and $\alpha < 65n$ (or $\alpha \leq 9n$ which is claimed in~\cite{AM17}). It is widely conjectured that the PoA is constant for all $\alpha$ and settling this open question is a long standing problem in the field. A constant PoA proves that agents create socially close-to-optimal networks even without central coordination. Quite recently, a constant PoA was proven by Chauhan et al.~\cite{CLMM17} for a version with non-uniform edge prices. In contrast, non-constant lower bounds on the PoA have been proven for local versions of the network creation game by Bilò et al.~\cite{BiloGLP14,BiloGLP16} and Cord-Landwehr \& Lenzner~\cite{CL15}.

For other variants and aspects of network creation games, we refer the reader to~\cite{JW96,BG00,CP05,BHN08,De09,AFM09,Kli11,L11,L12,MS12,KL13,ADHL13,MMO14,CMadH14,Ehs15,Bilo15,Bilo15_max,ABU15,GJKKM16,CLMM16,ABDMS16,AM16,JK17}.

\subsection{Our Contribution}
In this paper we introduce a new technique for analyzing stable non-tree networks for high edge-price $\alpha$ and use it to improve on the current best lower bound for $\alpha$ for which all stable networks must be trees. In particular, we prove that for $\alpha > 4n-13$ any stable network must be a tree (see Section~\ref{sec_treeconjecture}). This is a significant improvement over the known bound of $\alpha > 65n$ by Mamageishvili et al.~\cite{MMM13} and the recently claimed bound of $\alpha > 17n$ by Àlvarez \& Messegué~\cite{AM17}. Since the price of anarchy for stable tree networks is constant~\cite{Fab03}, our bound directly implies a constant price of anarchy for $\alpha > 4n-13$. Moreover, in Section~\ref{sec_PoA}, we also give a refined analysis of the price of anarchy of stable tree networks and thereby improve the best known constant upper bound for stable trees.     

Thus, we make significant progress towards settling the tree conjecture in network creation games and we enlarge the range of $\alpha$ for which the price of anarchy is provably constant.

Our new technique exploits properties of cycles in stable networks by focusing on \emph{critical pairs}, \emph{strong critical pairs} and \emph{min cycles}. The latter have been introduced in our earlier work~\cite{L14} and are also used in the preprint by Àlvarez \& Messegué~\cite{AM17}. However, in contrast to the last attempts for settling the tree conjecture~\cite{MS10,MMM13,AM17}, we do not rely on the average degree technique. Instead we propose a more direct and entirely structural approach using (strong) critical pairs in combination with min cycles. Besides giving better bounds with a simpler technique, we believe that this approach is better suited for finally resolving the tree conjecture. 

\section{Improving the Range of $\alpha$ of the Tree Conjecture}\label{sec_treeconjecture}
In this section we prove our main result, that is, we show that for $\alpha > 4n-13$, every NE network $(G,\alpha)$ with $n\geq 4$ nodes must be a tree. 

We proceed by first establishing properties of cycles in stable networks. Then we introduce the key concepts called critical pairs, strong critical pairs and min cycles. Finally, we provide the last ingredient, which is a critical pair with a specific additional property, and combine all ingredients to obtain the claimed result.
\subsection{Properties of Cycles in Stable Networks}
We begin by showing that for large values of $\alpha$, stable networks cannot contain cycles of length either 3 or 4.

\begin{lemma}\label{lm:no_3_cycle}
For $\alpha > \frac{n-1}{2}$, no stable network $(G,\alpha)$ contains a cycle of length 3.
\end{lemma}
\begin{proof}
Let $(G,\alpha)$ be a stable network for a fixed value of $\alpha > \frac{n-1}{2}$. For the sake of contradiction, assume that $G$ contains a cycle $C$ of length $3$. Assume that $V(C)=\{u_0,u_1,u_2\}$ and that $C$ contains the three edges $(u_0,u_1)$, $(u_1,u_2)$, and $(u_2,u_0)$. Let $V_i = \big\{x \in V \mid d_G(u_i,x) < d_G(u_j,x), \forall j \neq i\big\}$.
Observe that, for every $i\in\{0,1,2\}$ we have $|V_i|\geq 1$, as $u_i \in V_i$. Furthermore, all the $V_i$'s are pairwise disjoint. W.l.o.g., assume that $|V_{0}|= \max\big\{|V_0|,|V_1|,|V_2|\big\}$. Furthermore, w.l.o.g., assume that $u_1$ buys the edge $(u_1,u_2)$. Consider the strategy change in which agent $u_1$ deletes the edge $(u_1,u_2)$. The building cost of the agent decreases by $\alpha$ while her distance cost increases by at most $|V_2|$. Since $|V_2|\leq |V_0|$, from $|V_0|+|V_1|+|V_2| \leq n$ we obtain $|V_2|\leq \frac{n-1}{2}$. Since $G$ is stable, $\frac{n-1}{2}-\alpha \geq 0$, i.e., $\alpha\leq \frac{n-1}{2}$, a contradiction.
\end{proof}

\begin{lemma}\label{lm:no_4_cycle}
For $\alpha > n-2$, no stable network $(G,\alpha)$ contains a cycle of length~4.
\end{lemma}
\begin{proof}
Let $(G,\alpha)$ be a stable network for a fixed value of $\alpha > n-2$. For the sake of contradiction, assume that $G$ contains a cycle $C$ of length $4$. Assume that $V(C)=\{u_0,u_1,u_2,u_3\}$ and that $C$ contains the four edges $(u_0,u_1)$, $(u_1,u_2)$, $(u_2,u_3)$, and $(u_3,u_0)$. 
For the rest of this proof, we assume that all indices are modulo 4 in order to simplify notation.
Let $V_i = \big\{x \in V \mid d_G(u_i,x) < d_G(u_j,x), \forall j \neq i\big\}$.
Observe that for every $i\in\{0,1,2,3\}$ we have $|V_i|\geq 1$, as $u_i \in V_i$.
Let
$Z_i =\big\{ x \in V \mid d_G(u_i,x)=d_G(u_{i-1},x) \text{ and }
			 d_G(u_i,x),d_G(u_{i-1},x) < d_G(u_j,x), \forall j \neq i, i-1\big\}$.
Observe that in the families of the sets $V_i$ and $Z_i$ every pair of sets is pairwise disjoint.

We now rule out the case in which an agent owns two edges of $C$. W.l.o.g., assume that agent $u_0$ owns the two edges $(u_0,u_1)$ and $(u_0,u_3)$. Consider the strategy change in which agent $u_0$ swaps\footnote{A \emph{swap} of edge $(a,b)$ to edge $(a,c)$ by agent $a$ who owns edge $(a,b)$ consists of deleting edge $(a,b)$ and buying edge $(a,c)$.} the edge $(u_0,u_1)$ with the edge $(u_0,u_2)$ and, at the same time, deletes the edge $(u_0,u_3)$. The creation cost of agent $u_0$ decreases by $\alpha$, while her distance cost increases by $|V_1|+|V_3|-|V_2|$. Since $(G,\alpha)$ is stable, agent $u_0$ has no incentive in deviating from her current strategy. 
Therefore, $|V_1|+|V_3|-|V_2|-\alpha \geq 0$, i.e., $\alpha \leq |V_1|+|V_3|-|V_2| \leq n-|V_0|-|V_2| -|V_2| \leq  n-3$, where the last but one inequality follows from the pairwise disjointness of all $V_i$ sets, which implies $|V_0| + |V_1| + |V_2| + |V_3| \leq n$. 
Since, $\alpha >n-2$, no agent can own two edges of $C$. Therefore, to prove the claim, we need to show that no agent can own a single edge of $C$. 

W.l.o.g., assume that for every $i\in \{0,1,2,3\}$ agent $u_i$ owns the edge $(u_i,u_{i+1})$. Moreover, w.l.o.g., assume that
$|V_1|+|Z_2| = \min_{0\leq i \leq 3} \big\{|V_i|+|Z_{i+1}|\big\}$.
Since 
$\sum_{0\leq i\leq 3}\big(|V_i|+|Z_{i+1}|\big)\leq n$,
we have that 
$|V_1|+|Z_2|\leq \frac{n}{4}$.

Consider the strategy change in which agent $u_0$ deletes the edge $(u_0,u_1)$. The creation cost of agent $u_0$ decreases by $\alpha$, while her distance cost increases by $2|V_1|+|Z_2|\leq \frac{n}{2}$. 

Since $(G,\alpha)$ is stable, agent $u_0$ has no incentive to deviate from her current strategy. Therefore, $\frac{n}{2}-\alpha\geq 0$, i.e., $\alpha \leq \frac{n}{2} \leq n-2$, when $n\geq 4$. We have obtained a contradiction.
\end{proof}

\begin{definition}[Directed Cycle]\label{def:directed_cycle}
 Let $C$ be a cycle of $(G,\alpha)$ of length $k$. We say that $C$ is {\em directed} if there is an ordering $u_0,\dots,u_{k-1}$ of its $k$ vertices such that, for every $i=0,\dots,k-1$, $(u_i,u_{(i+1)\mod k})$ is an edge of $C$ which is bought by agent $u_i$.
\end{definition}
\noindent We now show that if $\alpha$ is large enough, then directed cycles cannot be contained in a stable network as a biconnected component.

\begin{lemma}\label{lm:no_directed_cycle_as_biconnected_component}
For $\alpha > n-2$, no stable network $(G,\alpha)$ with $n\geq 6$ vertices contains a biconnected component which is also a directed cycle.
\end{lemma}
\begin{proof}
Let $(G,\alpha)$ be a stable network for a fixed value of $\alpha > n-2$. Let $H$ be a biconnected component of $G$. For the sake of contradiction, assume that $H$ is a directed cycle of length $k$. We can apply Lemma~\ref{lm:no_4_cycle} to exclude the case in which $k=4$. Similarly, since $\alpha> n-2\geq \frac{n-1}{2}$ for every $n\geq 3$, we can use Lemma~\ref{lm:no_3_cycle} to exclude the  case in which $k=3$.

Let $u_0,\dots,u_{k-1}$ be the $k$ vertices of $H$ and, w.l.o.g., assume that every agent $u_i$ is buying an edge towards agent $u_{(i+1)\mod k}$. To simplify notation, in the rest of this proof we assume that all indices are modulo $k$. Let
$V_i = \big\{x \in V \mid d_G(u_i,x) < d_G(u_j,x), \forall j \neq i\big\}$.
Observe that $V_i$ is a partition of $V$. We divide the proof into two cases.

The first case occurs when $H$ is a cycle of length $k\geq 6$. W.l.o.g., assume that $|V_2| = \max_{0\leq i \leq k-1}|V_i|$. In this case, consider the strategy change of agent  $u_0$ when she swaps the edge $(u_0,u_1)$ with the edge $(u_0,u_2)$. The distance cost of agent $u_0$ increases by $|V_1|-|V_2|-|V_3|\leq -1$. Thus, agent $u_0$ has an improving strategy, a contradiction.

The second and last case occurs when $H$ is a cycle of length $k=5$. If $|V_i|\neq |V_j|$ for some $i,j\in\{0,1,2,3,4\}$, then there exists an $i \in \{0,1,2,3,4\}$ such that $|V_i|<|V_{i+1}|$. W.l.o.g., let $|V_1|<|V_2|$. The distance cost of agent $u_0$ when she swaps the edge $(u_0,u_1)$ with the edge $(u_0,u_2)$ increases by $|V_1|-|V_2|\leq -1$. Thus agent $u_0$ has an improving strategy, a contradiction. If $|V_0|=|V_1|=|V_2|=|V_3|=|V_4|$, then the increase in the overall cost incurred by agent $u_0$ when she deletes the edge $(u_0,u_1)$ would be equal to $3|V_1|+|V_2|-\alpha=\frac{4}{5}n-\alpha$. Since $G$ is stable and $n$ is a multiple of $5$, $\frac{4}{5}n-\alpha \geq 0$, i.e., $\alpha\leq \frac{4}{5}n\leq n-2$, for every $n\geq 10$, a contradiction. 
\end{proof}
\subsection{Critical Pairs}
\noindent The next definition introduces the concept of a (strong) critical pair. As we will see, (strong) critical pairs are the first key ingredient for our analysis. Essentially, we will show that stable networks cannot have critical pairs, if $\alpha$ is large enough.

\begin{definition}[Critical Pair]\label{definition:apx_critical_pair}
Let $(G,\alpha)$ be a non-tree network and let $H$ be a biconnected component of $G$.  
We say that $\langle v, u \rangle$ is a {\em critical pair} if all of the following five properties hold:
\begin{enumerate}
\item Agent $v \in V(H)$ buys two distinct non-bridge edges, say $(v,v_1)$ and $(v,v_2)$, with $v_1,v_2 \in V(H)$;
\item Agent $u \in V(H)$, with $u\neq v$ buys at least one edge $(u,u')$ with $u' \in V(H)$ and $u'\neq v$;
\item $d_G(v,u) \geq 2$;
\item there is a shortest path between $v$ and $u$ in $G$ which uses the edge $(v,v_1)$;
\item there is a shortest path between $v$ and  $u'$ in $G$ which does not use the edge $(u,u')$.
\end{enumerate}
The critical pair $\langle v, u\rangle$ is {\em strong} if there is a shortest path between $u$ and $v_2$ which does not use the edge $(v,v_2)$. See Fig.~\ref{fig:critical_pair} for an illustration.
\end{definition}
\begin{figure}[!htb]
 \centering
 \includegraphics[width=\textwidth]{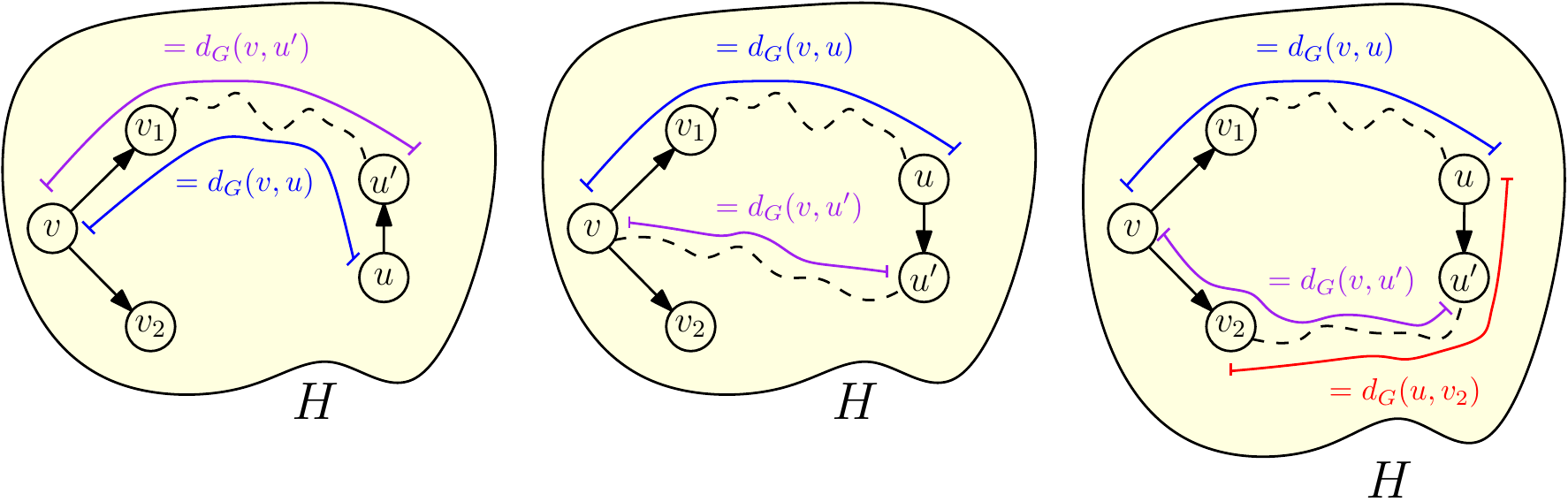}
 \caption{Illustrations of a critical pair $\langle v,u\rangle$. Edge-ownership is depicted by directing edges away from their owner. \textbf{Left:} Edge $(u,u')$ belongs to a shortest path tree $T$ rooted at $v$ and $u'$ is the parent of $u$ in $T$. \textbf{Middle:} Edge $(u,u')$ does not belong to any shortest path tree $T$ rooted at $v$. Note that in this case $(v,v_2)$ can also be on the shortest path from $v$ to $u'$. \textbf{Right:} Illustration of a strong critical pair $\langle v,u\rangle$.}
 \label{fig:critical_pair}
\end{figure}

\noindent In the rest of this section, when we say that two vertices $v$ and $u$ of $G$ form a critical pair, we will denote by $v_1,v_2$, and $u'$ the vertices corresponding to the critical pair $\langle v,u\rangle$ that satisfy all the conditions given in Definition~\ref{definition:apx_critical_pair}. We can observe the following.

\begin{observation}\label{remark:u}
If $\langle v, u\rangle$ is a critical pair of a network $(G,\alpha)$, then there exists a shortest path tree $T$ of $G$ rooted at $v$, where either the edge $(u,u')$ is not an edge of $T$ or $u'$ is the parent of $u$ in $T$.
\end{observation}
\begin{observation}\label{remark:v}
If $\langle v, u\rangle$ is a critical pair, then for every shortest path tree $T$ of $(G,\alpha)$ rooted at $u$, either the edge $(v,v_1)$ is not an edge of $T$ or $v_1$ is the parent of $v$ in $T$. Furthermore, if $\langle v,u\rangle$ is a strong critical pair, then there is a shortest path tree of $G$ rooted at $u$ such that the edge $(v,v_2)$ is not contained in the shortest path tree. 
\end{observation}
\noindent The next technical lemma provides useful bounds on the distance cost of the nodes involved in a critical pair.

\begin{lemma}\label{lemma:edge_pointing_towards_the_root_or_not_in_the_shortest_path}
Let $(G,\alpha)$ be a stable network and let $a,b$ be two distinct vertices of $G$ such that $a$ buys an edge $(a,a')$, with $a'\neq b$. If $d_G(a,b)\geq 2$ and there exists a shortest path tree $T$ of $G$ rooted at $b$ such that either $(a,a')$ is not an edge of $T$ or $a'$ is the parent of $a$ in $T$, then $\cost(a)\leq \cost(b)+n-3$. 
Furthermore, if $a$ is buying also the edge $(a,a'')$, with $a''\neq a'$, $a'' \neq b$, and $(a,a'')$ is not an edge of $T$, then $\cost(a) \leq \cost(b)+n-3-\alpha$.
\end{lemma}
\begin{proof}
Consider the strategy change in which agent $a$ swaps the edge $(a,a')$ with the edge $(a,b)$ and deletes any other edge she owns and which is not contained in $T$, if any. 
Let $T'$ be a shortest path tree rooted at $b$ of the graph obtained after the swap. 
Observe that $d_{T'}(b,x)\leq d_{G}(b,x)$, for every $x \in V$. Furthermore, as $d_G(a,b)\geq 2$, while $d_{T'}(a,b)=1$, we have $d_{T'}(a,b)\leq d_G(a,b)-1$. Therefore, $\sum_{x \in V}d_{T'}(b,x)\leq \cost(b)-1$. Moreover, the distance from $a$ to every $x\neq a$ is at most $1+d_{T'}(b,x)$. Finally, the distance from $a$ to herself, which is clearly 0, is exactly 1 less than the distance from $b$ to $a$ in $T'$. Therefore the distance cost of $a$ in $T'$ is less than or equal to $\cost(b)-1+(n-1)-1=\cost(b)+n-3$. 

If besides performing the mentioned swap agent $a$ additionally saves at least $\alpha$ in cost by deleting at least one additional edge which is not in $T$, then $\cost(a)\leq \cost(b)+n-3-\alpha$. This is true since $G$ is stable, which implies that the overall cost of $a$ in $G$ cannot be larger than the overall cost of $a$ after the strategy change.
\end{proof}
\noindent Now we employ Lemma~\ref{lemma:edge_pointing_towards_the_root_or_not_in_the_shortest_path} to prove the structural property that stable networks cannot contain strong critical pairs if $\alpha$ is large enough.

\begin{lemma}\label{lemma:no_strong_critical_pair}
For $\alpha > 2n-6$, no stable network $(G,\alpha)$ contains a strong critical pair.
\end{lemma}
\begin{proof}
Let $(G,\alpha)$ be a non-tree stable network for a fixed value of $\alpha > 2n-6$ and, for the sake of contradiction, let $\langle v,u  \rangle$ be a strong critical pair. Using Observation~\ref{remark:u} together with Lemma~\ref{lemma:edge_pointing_towards_the_root_or_not_in_the_shortest_path} (where $a=u, a'=u'$, and $b=v$), we have that 
\begin{equation*}
\cost(u)\leq \cost(v)+n-3.
\end{equation*}
Furthermore, using Observation~\ref{remark:v} together with Lemma~\ref{lemma:edge_pointing_towards_the_root_or_not_in_the_shortest_path} (where $a=v, a'=v_1, a''=v_2$, and $b=u$), we have that 
\begin{equation*}
\cost(v)\leq \cost(u)+n-3-\alpha.
\end{equation*}
By summing up both the left-hand and the right-hand side of the two inequalities we obtain $0\leq 2n-6-\alpha$, i.e., $\alpha\leq 2n-6$, a contradiction.
\end{proof}
\subsection{Min Cycles}
\noindent We now introduce the second key ingredient for our analysis: min cycles.

\begin{definition}[Min Cycle]
Let $(G,\alpha)$ be a non-tree network and let $C$ be a cycle in $G$. We say that $C$ is a min cycle if, for every two vertices $x,x' \in V(C)$, $d_C(x,x')=d_{G}(x,x')$.
\end{definition}

\noindent First, we show that every edge of every biconnected graph is contained in some min cycle. This was also proven in~\cite{L14} and \cite{AM17}.

\begin{lemma}\label{lemma:existence_of_min_cycle}
Let $H$ be a biconnected graph. Then, for every edge $e$ of $H$, there is a min cycle that contains the edge $e$.
\end{lemma}
\begin{proof}
Since $H$ is biconnected, there exists at least a cycle containing the edge $e$. 
Among all the cycles in $H$ that contain the edge $e$, let $C$ be a cycle of minimum length. We claim that $C$ is a min cycle. For the sake of contradiction, assume that $C$ is not a min cycle. This implies that there are two vertices $x,y \in V(C)$ such that $d_H(x,y) < d_C(x,y)$. Among all pairs $x,y \in V(C)$ of vertices such that $d_H(x,y)<d_C(x,y)$, let $x',y'$ be the one that minimizes the value $d_H(x',y')$ (ties are broken arbitrarily). Let $\Pi$ be a shortest path between $x'$ and $y'$ in $G$. By the choice of $x'$ and $y'$, $\Pi$ is edge disjoint from $C$.
\begin{figure}[h]
 \centering
 \includegraphics[width=6cm]{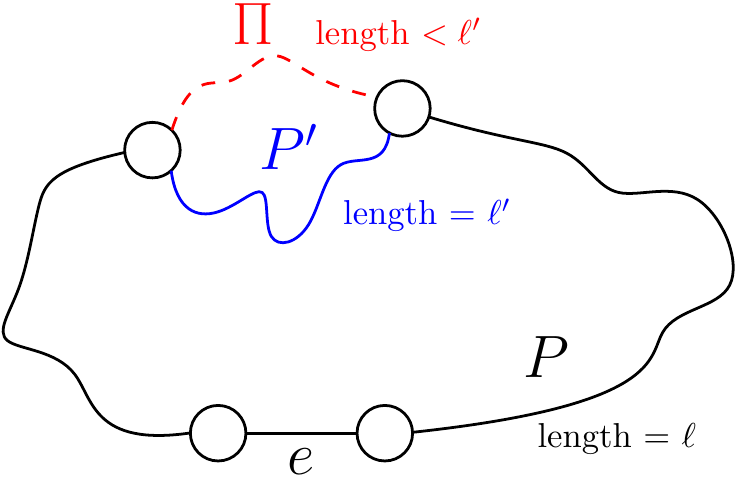}
 \caption{The cycle $C$ containing edge $e$ and the paths $P,P'$ and $\Pi$.}
 \label{fig:min_cycle}
\end{figure}
Let $P$ and $P'$ be the two edge-disjoint paths between $x'$ and $y'$ in $C$ and, w.l.o.g., assume that $e$ is contained in $P$. Let $\ell$ and $\ell'$ be the length of $P$ and $P'$, respectively. See Fig.~\ref{fig:min_cycle}. Clearly, the length of $C$ is equal to $\ell + \ell'$. Since $d_C(x',y') \leq \ell'$, we obtain $d_H(x',y') < \ell'$. Therefore, the cycle obtained by concatenating $P$ and $\Pi$ has a length equal to $\ell+d_H(x',y')< \ell+\ell'$, and therefore, it is strictly shorter than $C$, a contradiction.
\end{proof}

\noindent Now we proceed with showing that stable networks contain only min cycles which are directed and not too short. For this, we employ our knowledge about strong critical pairs.

\begin{lemma}\label{lm:every_min_cycle_is_directed_and_has_length_at_least_5}
For $\alpha > 2n-6$, every min cycle of a non-tree stable network $(G,\alpha)$ with $n\geq 4$ vertices is directed and has a length of at least 5.
\end{lemma}
\begin{proof}
Let $(G,\alpha)$ be a non-tree stable network for a fixed $\alpha > 2n-6$ and let $C$ be a min cycle of $G$. 
Since $2n-6\geq \frac{n-1}{2}$ for every $n\geq 4$, using Lemma~\ref{lm:no_3_cycle}, we have that $C$ cannot be a cycle of length equal to 3. Furthermore, Since $2n-6\geq n-2$ for every $n\geq 4$, using Lemma~\ref{lm:no_4_cycle}, we have that $C$ cannot be either a cycle of length equal to 4. Therefore, $C$ is a cycle of length greater than or equal to 5.

For the sake of contradiction, assume that $C$ is not directed. This means that $C$ contains a agent, say $v$, that is buying both her incident edges in $C$. We prove the contradiction thanks to Lemma~\ref{lemma:no_strong_critical_pair}, by showing that $C$ contains a strong critical pair.

If $C$ is an odd-length cycle, then $v$ has two distinct antipodal vertices $u, u' \in V(C)$ which are also adjacent in $C$.\footnote{In a cycle of length $\ell$, two vertices of the cycles are {\em antipodal} if their distance is equal to $\lfloor \ell /2 \rfloor$.} W.l.o.g., assume that $u$ is buying the edge towards $u'$. Clearly, $d_G(v,u)\geq 2$. Furthermore, since $C$ is a min cycle, it is easy to check that $\langle v, u \rangle$ is a strong critical pair.

If $C$ is an even-length cycle, then let $u \in V(C)$ be the (unique) antipodal vertex of $v$  and let $u'$ be a vertex that is adjacent to $u$ in $C$. Observe that $d_G(v,u),d_G(v,u')\geq 2$. Again using the fact that $C$ is a min cycle, we have the following:
\begin{itemize}
\item If $u$ is buying the edge towards $u'$, then $\langle v, u \rangle$ is a strong critical pair. 
\item If $u'$ is buying the edge towards $u$, then $\langle v, u' \rangle$ is a strong critical pair.
\end{itemize}
In both cases, we have proved that $C$ contains a strong critical pair. 
\end{proof}
\noindent Let $(G,\alpha)$ be a non-tree stable network with $n\geq 6$ vertices for a fixed $\alpha > 2n-6$ and let $H$ be a biconnected component of $G$. Since $2n-6\geq n-2$ for every $n\geq 4$, Lemma~\ref{lm:no_directed_cycle_as_biconnected_component} implies that $H$ cannot be a directed cycle. At the same time, if $H$ is a cycle, then it is also a min cycle and therefore, Lemma~\ref{lm:every_min_cycle_is_directed_and_has_length_at_least_5} implies that $H$ must be directed, which contradicts Lemma~\ref{lm:no_directed_cycle_as_biconnected_component}. Therefore, we have proved the following.

\begin{corollary}\label{cor:no_cycle_as_biconnected_component}
For $\alpha > 2n-6$, no non-tree stable network $(G,\alpha)$ with $n\geq 6$ vertices contains a cycle as one of its biconnected components.
\end{corollary}
\subsection{Combining the Ingredients}
\noindent Towards our main result, we start with proving that every stable network must contain a critical pair which satisfies an interesting structural property. This lemma is the third and last ingredient that is used in our analysis.

\begin{lemma}\label{lemma:cool_path}
For $\alpha > 2n - 6$, every non-tree stable network $(G,\alpha)$ with $n\geq 6$ vertices contains a critical pair $\langle v, u \rangle$. Furthermore, there exists a path $P$ between $v$ and $v_2$ in $G$ such that (a) the length of $P$ is at most $2d_G(u,v)$ and (b) $P$ uses none of the edges $(v,v_1)$ and $(v,v_2)$. 
\end{lemma}
\begin{proof}
Let $(G,\alpha)$ be a network of $n\geq 6$ vertices which is stable for a fixed $\alpha > 2n-6$, and let $H$ be any biconnected component of $G$. By Corollary~\ref{cor:no_cycle_as_biconnected_component}, we have that $H$ cannot be a cycle. As a consequence, $H$ contains at least $|V(H)|+1$ edges and, therefore, it has a vertex, say $v$, that buys at least two edges of $H$.

Let $v_1$ and $v_2$ be the two distinct vertices of $H$ such that $v$ buys the edges $(v,v_1)$ and $(v,v_2)$. Let $C_i$ be the min cycle that contains the edge $(v,v_i)$, whose existence is guaranteed by Lemma~\ref{lemma:existence_of_min_cycle}. Lemma~\ref{lm:every_min_cycle_is_directed_and_has_length_at_least_5} implies that $C_i$ is a directed cycle of length greater than or equal to 5. Therefore, since $(v,v_1)$ is an edge of $C_1$ bought by agent $v$, $C_1$ cannot contain the edge $(v,v_2)$, which is also bought by $v$. Similarly, since $(v,v_2)$ is an edge of $C_2$ bought by agent $v$, $C_2$ cannot contain the edge $(v,v_1)$, which is also bought by $v$.

Let $T$ be a shortest path tree rooted at $v$ which gives priority to the shortest paths using the edges $(v,v_1)$ or $(v,v_2)$. More precisely, for every vertex $x$, if there is a shortest path from $v$ to $x$ containing the edge $(v,v_1)$, then $x$ is a descendant of $v_1$ in $T$. Furthermore, if no shortest path from $v$ to $x$ contains the edge $(v,v_1)$, but there is a shortest path from $v$ to $x$ containing the edge $(v,v_2)$, then $x$ is a descendant of $v_2$ in $T$.

Consider the directed version of $C_i$ in which each edge is directed from their owner agent towards the other end vertex. 
Let $u_i$ be, among the vertices of $C_i$ which are also descendants of $v_i$ in $T$, the one which is in maximum distance from $v$ w.r.t. the directed version of $C_i$. 
Finally, let $(u_i,u_i')$ be the edge of $C_i$ which is bought by agent $u_i$. Clearly, $u_i'$ is not a descendant of $v_i$ in $T$. Therefore, by construction of $T$, $d_G(v,u_i')\leq d_G(v,u_i)$, otherwise $u_i'$ would have been a descendant of $v_i$ in $T$, or there would have been a min cycle containing both edges $(v,v_1)$ and $(v,v_2)$ (which are both bought by agent $v$), thus contradicting Lemma~\ref{lm:every_min_cycle_is_directed_and_has_length_at_least_5}.

W.l.o.g., assume that $d_G(v,u_2)\leq d_G(v,u_1)$. Let $u=u_1$ and $u'=u_1'$.
We show that $\langle v, u \rangle$ is a critical pair. By Lemma~\ref{lm:every_min_cycle_is_directed_and_has_length_at_least_5}, $C_1$ is a cycle of length $k\geq 5$. As $C_1$ is a min cycle, $k=d_G(v,u)+1+ d_G(v,u')$. Moreover, since $d_G(v,u')\leq d_G(v,u)$, we have that $d_G(v,u) \geq \frac{k-1}{2}\geq 2$. Therefore $u'\neq v$. Next, the shortest path in $T$ between $v$ and $u$ uses the edge $(v,v_1)$ which is owned by agent $v$. Furthermore, the shortest path in $T$ between $v$ and $u'$ does not use the edge $(u,u')$. Therefore, $\langle v, u \rangle$ is a critical pair.

Now, consider the path $P$ which is obtained from $C_2$ by removing the edge $(v,v_2)$. Recalling that $C_2$ does not contain the edge $(v,v_1)$, it follows that $P$ is a path between $v$ and $v_2$ which uses none of the two edges $(v,v_1)$ and $(v,v_2)$. Therefore, recalling that $d_G(v,u_2')\leq d_G(v,u_2)$, the overall length of $P$ is less than or equal to
\begin{align*}
d_G(v,u_2')+1+d_G(v_2,u_2) & \leq d_G(v,u_2)+1+d_G(v,u_2)-1 
							 \leq 2d_G(v,u_1)=2d_G(v,u).\qedhere
\end{align*}
\end{proof}
\noindent Finally, we prove our main result. For this and in the rest of the paper, given a vertex $x$ of a network $(G,\alpha)$ and a subset $U$ of vertices of $G$, we denote by $d_{G}(x,U) := \sum_{x' \in U}d_{G}(x,x')$. 

\begin{theorem}\label{thm_tree}
For $\alpha > 4n-13$, every stable network $(G,\alpha)$ with $n\geq 4$ vertices is a tree.
\end{theorem}
\begin{proof}
 First of all, it is easy to check that for $\alpha > 3$ every stable network with $n=4$ vertices is a tree. Moreover, the same holds true for $n=5$ for $\alpha > 7$. 
 
 Let $\alpha >4n-13$ be a fixed value and let $(G,\alpha)$ be a stable network with $n\geq 6$ vertices. Since $4n-13\geq 2n-6$, for every $n\geq 4$, we have that if $(G,\alpha)$ is not a tree, then, by Lemma~\ref{lemma:cool_path}, it contains a critical pair $\langle v, u\rangle$ satisfying the conditions stated in Lemma~\ref{lemma:cool_path}. Moreover, Lemma~\ref{lemma:no_strong_critical_pair} implies that $\langle v, u\rangle$ cannot be a strong critical pair. As a consequence, every shortest path from $u$ to $v_2$ uses the edge $(v,v_2)$. Since $\langle v,u \rangle$ is a critical pair, this implies that there is a shortest path from $u$ to $v_2$ which uses both the edges $(v_1,v)$ and $(v,v_2)$. To finish our proof, we show that this contradicts the assumed stability of $(G,\alpha)$. This implies that $(G,\alpha)$ must be a tree.

Let $T(u)$ be a shortest path tree of $G$ rooted at $u$ having $v_1$ as the parent of $v$ and $v$ as the parent of $v_2$. Observe that, by definition of a critical pair, there is a shortest path between $v$ and $u$ containing the edge $(v,v_1)$. Therefore, $T(u)$ is well defined. Furthermore, let $X$ be the set of vertices which are descendants of $v_2$ in $T(u)$. Note that since $v_2 \in X$, we have $|X| \geq 1$.

Since $\langle v, u \rangle$ is a critical pair, thanks to Observation~\ref{remark:u}, we can use Lemma~\ref{lemma:edge_pointing_towards_the_root_or_not_in_the_shortest_path} (where $a=u, a'=u'$, and $b=v$) to obtain
\begin{equation}\label{equation:usage_cost_of_u}
\cost(u)\leq {\cost}(v)+n-3.
\end{equation}
Furthermore, observe that
\begin{align}\label{equation:alternative_usage_cost_of_u}
{\cost}(u) 	& = \sum_{x \in X}\big(d_G(u,v)+d_G(v,x)\big) + d_G(u,V\setminus X)\\
			& = d_G(u,v)|X| + d_G(v,X) + d_G(u,V\setminus X).\nonumber
\end{align}
Therefore, by substituting $\cost(u)$ in~(\ref{equation:usage_cost_of_u}) with~(\ref{equation:alternative_usage_cost_of_u}) we obtain the following
\begin{equation}\label{equation:delta_of_u}
d_G(u,v)|X| + d_G(v,X) + d_G(u,V\setminus X) \leq {\cost}(v) + n-3.
\end{equation}
Let $T'(u)$ be the tree obtained from $T(u)$ by the the swap of the edge $(v,v_1)$ with the edge $(v,u)$. The distance cost incurred by agent $v$ if she swaps the edge $(v,v_1)$ with the edge $(v,u)$ is at most 
\begin{align*}
d_{T'(u)}(v,V) 	& = d_{T'(u)}(v,X) + d_{T'(u)}(v,V\setminus X)
 = d_{T(u)}(v,X) + d_{T'(u)}(v,V\setminus X)\\
								& \leq  d_{T(u)}(v,X) + \sum_{x \in V\setminus (X\cup \{v\})}\big(1+d_{T(u)}(u,x)\big) \\ 
								& \leq d_{G}(v,X) + \sum_{x \in V\setminus X}\big(1+d_{G}(u,x)\big)-2\\ 
								& = d_G(v,X) + n-|X| + d_G(u,V\setminus X)-2.
\end{align*}
Since $(G,\alpha)$ is stable, agent $v$ cannot decrease her distance cost by swapping any of the edges she owns.
Therefore, we obtain
\begin{equation}\label{equation:delta_of_v}
{\cost}(v) \leq d_G(v,X) + n-|X| + d_G(u,V\setminus X)-2.
\end{equation}
By summing both the left-hand and the right-hand sides of the two inequalities~(\ref{equation:delta_of_u}) to~(\ref{equation:delta_of_v}) and simplifying we obtain
\begin{equation}\label{equation:nice_inequality}
d_G(u,v)|X| \leq 2n - 5 - |X|.
\end{equation}
Consider the network $(G',\alpha)$ induced by the strategy vector in which agent $v$ deviates from her current strategy by swapping the edge $(v,v_1)$ with the edge $(v,u)$ and, at the same time, by deleting the edge$(v,v_2)$.
By Lemma~\ref{lemma:cool_path}, there exists a path $P$ between $v$ and $v_2$ in $G$, of length at most $2d_G(u,v)$, such that  $P$ uses none of the edges $(v,v_1)$ and $(v,v_2)$. 
As a consequence, using both~(\ref{equation:usage_cost_of_u}) and~(\ref{equation:nice_inequality}) in the second to last inequality of the following chain, the distance cost of $v$ w.r.t. $(G',\alpha)$ is upper bounded by 
\begin{align*}
d_{G'}(v,V) 	& \leq \sum_{x \in X} \big(2d_G(u,v)+d_G(v_2,x)\big) + \sum_{x \in V\setminus (X\cup\{v\})}\big(1+d_G(u,x)\big)\\
				& \leq 2d_G(u,v)|X|+d_G(v_2,X)+n-|X|+d_G(u,V\setminus X)-2\\
				& \leq 2d_G(u,v)|X| + d_G(v,X) - |X|  + n-|X| + d_G(u,V\setminus X)-2\\
				& = 2d_G(u,v)|X| + d_G(u,X)-d_G(u,v)|X|+n-2|X|+d_G(u,V\setminus X)-2\\
				& = d_G(u,v)|X| + n - 2|X| + {\cost}(u)-2\\
				& \leq 2n-5-|X| + n - 2|X| + {\cost}(v) + n - 3-2\\
				& =  {\cost}(v) + 4n - 10 - 3|X|
				\leq {\cost}(v) + 4n -13.
\end{align*}
By her strategy change, agent $v$ will save $\alpha$ in edge cost and her distance cost will increase by at most $4n-13$. Thus, if $\alpha > 4n-13$, then this yields a strict cost decrease for agent $v$ which contradicts the stability of $(G,\alpha)$.
\end{proof}

\noindent With the results from Fabrikant et al.~\cite{Fab03} Theorem~\ref{thm_tree} yields:
\begin{corollary}
For $\alpha > 4n-13$ the PoA is at most $5$.
\end{corollary}
\noindent In Section~\ref{sec_PoA} we improve the upper bound of $5$ on the PoA for stable tree networks from Fabrikant et al.~\cite{Fab03}. With this, we establish the following:
\begin{corollary}
 For every $\alpha > 4n-13$ the PoA is at most $3+\frac{2n}{2n+\alpha}$. 
\end{corollary}

\section{Improved Price of Anarchy for Stable Tree Networks}\label{sec_PoA}

In this section we show a better bound on the PoA of stable tree networks. To prove the bound, we need to introduce some new notation first. Let $T$ be a tree on $n$ vertices and, for a vertex $v$ of $T$, let $T-v$ be the forest obtained by removing vertex $v$ together with all its incident edges from $T$. We say that $v$ is a {\em centroid} of $T$ if every tree in $T-v$ has at most $n/2$ vertices. It is well known that every tree has at least one centroid vertex~\cite{KH79}.

\begin{lemma}\label{lemma:centroid}
Let $(T,\alpha)$ be a stable tree network rooted at a centroid $c$ of $T$, and let $u,v \in V(T)$, with $u,v \neq c$, be two vertices such that $u$ buys the edge towards $v$ in $T$. Then $d_T(c,u) < d_T(c,v)$, i.e. $u$ is the parent of $v$ in $T$. Furthermore, if $\overline{T}$ denotes the subtree of $T$ rooted at $v$, then  $v$ is a centroid of~$\overline{T}$.
\end{lemma}
\begin{proof}
We show that $d_T(c,u) < d_T(c,v)$ by proving that if $d_T(c,u)>d_T(c,v)$, then $(T,\alpha)$ is not stable. So, assume that $d_T(c,u)>d_T(c,v)$. Consider the forest obtained from $T$ after the removal of the edge $(u,v)$ and let $T_v$ be the tree of the forest that contains vertex $v$. Since $v$ is closer to $c$ than $u$ in $T$ and $u$ is not a vertex of $T_v$, it follows that $c$ is a vertex of $T_v$. 
Consider the strategy change in which player $u$ swaps the edge $(u,v)$ with the edge $(u,c)$. Since $u$ and $v$ are both in the same tree, say $T''$, of the forest $T-c$, it follows that the tree induced by all the vertices of $T$ which are not contained in $T''$, say $T'$, is entirely contained in $T_v$ and has $n'\geq n/2$ vertices, as $c$ is a centroid of $T$. Observe that  after the swap of the edge $(u,v)$ with the edge $(u,c)$, the distance from each of the vertices in $T'$ decreases by $d_T(v,c)$, while the distance from each of all the other vertices of $T_v$ increases by at most $d_T(v,c)$. 
Therefore, if we denote by $n_v$ the number of vertices of $T_v$, then the usage cost of player $u$ increases by at most
\[
d_T(v,c)(n_v-n')-d_T(v,c)n'= d_T(v,c)(n_v-2n')\leq d_T(v,c)(n-1-2n/2) \leq -1.
\]
Therefore, $(T,\alpha)$ is not stable.

We now prove that $v$ is a centroid of $\overline{T}$. Let $\overline{V}$ be the set of vertices of $\overline{T}$. Observe that the claim trivially holds if $|\overline{V}|\leq 2$. Therefore, we assume that $|\overline{V}|\geq 3$. Notice that
\[
\cost(u) = d_T(u,\overline{V})+d_T(u,V\setminus \overline{V}) = |\overline{V}|+d_T(v,\overline{V})+d_T(u,V \setminus \overline{V}).
\]
Since $(T,\alpha)$ is stable, $d_T(v,\overline{V})=\min_{v' \in \overline{V}}d_T(v',\overline{V})$, otherwise $u$ would have incentive to change her strategy by swapping the edge $(u,v)$ with the edge $(u,v^*)$ such that $v^* \in \arg\min_{v' \in \overline{V}}d_T(v',\overline{V})$. 

Let $x_1,\dots,x_k$ be the $k$ neighbors of $v$ in $\overline{T}$. Clearly, $x_1,\dots,x_k$ are also the $k$ children of $v$ in $T$. Let $(v,x_i)$ be any edge of $\overline{T}$ adjacent to $v$. Consider the forest $F$ obtained by removing the edge $(v,x_i)$ from $\overline{T}$. Let $X_i$ be the set of vertices of the tree of $F$ that contains $x_i$. Let $Y_i=\overline{V}\setminus X_i$ be the set of vertices of the tree of $F$ that contains $v$. We have that 
\[
d_T(v, \overline{V}) = d_T(v,Y_i)+d_T(v,X_i)=d_T(v,Y_i)+|X_i|+d_T(x_i,X_i).
\]
Similarly,
\[
d_T(x_i,\overline{V})=d_T(x_i,Y_i)+d_T(x_i,X_i)=|Y_i|+d_T(v,Y_i)+d_T(x_i,X_i).
\]
Since $d_T(v,\overline{V})\leq d_T(x_i,\overline{V})$, it follows that
\[
d_T(v,Y_i)+|X_i|+d_T(x_i,X_i) \leq |Y_i|+d_T(v,Y_i)+d_T(x_i,X_i) \text{, i.e., } |X_i|\leq |Y_i|.
\]
Therefore, for every $i=1,\dots,k$, we have that $|X_i| \leq 1+\sum_{j=1,j\neq i}^{k}|X_j| = |\overline{V}|-|X_i|$, which implies that $|X_i| \leq |\overline{V}|/2$. Hence, $v$ is a centroid of $\overline{T}$.
\end{proof}
\noindent We now show a useful bound on the number of vertices contained in each of the subtrees of a stable tree network rooted at a centroid.

\begin{lemma}\label{lemma:lower_bound_vertices_subtrees}
Let $(T,\alpha)$ be a stable tree network rooted at a centroid $c$ of $T$, let $u$ be a child of $c$ in $T$ and let $v$ be a leaf of $T$ contained in the subtree of $T$ rooted at $u$. Let $c_1,\dots,c_k$ be the vertices along the path in $T$ between $c_0=u$ and $c_k=v$, where $c_{i+1}$ is the child of $c_i$, and, finally, for every $i=1,\dots,k$, let 
\[
n_i=\big|\big\{x \in V \mid d_T(c_i,x)<d_T(c_j,x), j\neq i\big\}\big|.
\]
We have that $\sum_{j=1}^{i}n_j \geq n\cdot\sum_{j=1}^i 1/{2^j}$.
\end{lemma}
\begin{proof}
The proof is by induction on $i$.

The base is when $i=1$. Observe that $\{x \in V \mid d_T(c,x)<d_T(c_j,x), j \in \{1,\dots,k\}\} \subset \{x \in V \mid d_T(c_1,x)<d_T(c_j,x), j \in \{2,\dots,k\}\}$. Therefore, since $c$ is a centroid, $n_1 \geq n/2$.

Now, assume that for every $j\leq i$, $n_j \geq n/2^j$. We prove the claim for $i+1$. Using Lemma~\ref{lemma:centroid}, we have that $c_{i+1}$ is a centroid of the subtree of $T$ rooted at $c_{i+1}$. As a consequence, if $m=\sum_{j=1}^{i} n_j$, we have that $n_{i+1} \geq \frac{n-m}{2}$.
By induction hypothesis, $m \geq n\cdot \sum_{j=1}^i 1/2^j$. Therefore,
\[
\sum_{j=1}^{i+1} n_j = m+n_{i+1} \geq m+\frac{n-m}{2}=\frac{m+n}{2}\geq n\cdot \sum_{j=1}^i 1/2^{j+1}+\frac{n}{2}=n\cdot \sum_{j=1}^{i+1}\frac{1}{2^j}. \qedhere
\]
\end{proof}

\noindent We can finally prove our upper bound on the PoA of stable tree networks.

\begin{theorem}\label{theorem:PoA_stable_trees}
For $\alpha \geq 2$, the PoA restricted to the class of stable tree networks of $n$ vertices is upper bounded by $3+\frac{2n^2-8n-4\alpha}{2n^2+(\alpha-2)n}$. 
\end{theorem}
\begin{proof}
Let $(T,\alpha)$ be a stable tree network rooted at a centroid $c$ of $T$. Let $c'$ be any child of $c$ in $T$ and let $v$ be any leaf contained in the subtree of $T$ rooted at $c'$. Let $c_1,\dots,c_{r+2}$ be the vertices along the path in $T$ between $c_1=c'$ and $c_{r+2}=v$, where $c_{i+1}$ is the child of $c_i$, and, finally, for every $i=1,\dots,r+2$, let 
\[
n_i=\big|\big\{x \in V \mid d_T(c_i,x)<d_T(c_j,x), j\neq i\big\}\big|.
\]
Consider the strategy change in which player $v$ buys the edge $(v,c_1)$ and let $k=\lfloor r/2 \rfloor$. The creation cost of player $v$ clearly increases by $\alpha$, while her distance cost decreases by at least $\sum_{i=1}^{k} \big((r+2-2i) n_{i}\big)$. Since, $r+2-2i$ is strictly positive and monotonically decreasing w.r.t. $i$, using Lemma~\ref{lemma:lower_bound_vertices_subtrees} we can observe that the distance cost of player $v$ is minimized when, for every $i=1,\dots,k$, $n_{r-i}$ is minimum, i.e., when $n_{r-i} \geq \frac{n}{2^{i}}$. Therefore, the distance cost decrease of $v$ is lower bounded by \vspace*{-0.2cm}
\begin{align*}
\sum_{i=1}^{k} \big((r+2-2i) n_{i}\big) 	& \geq \sum_{i=1}^{k} \big((r+2-2i) n/2^i\big) = (r+2)n\sum_{i=1}^{k}\frac{1}{2^i}-2n\sum_{i=1}^k \frac{i}{2^i}\\
											& = (r+2)n\left (1-\frac{1}{2^k}\right)-2n\left(\frac{k/2^{k+2}-(k+1)/2^{k+1}+1/2}{(1/2-1)^2}\right)\\
											& = (r+2)n-\frac{r+2}{2^k}n - \frac{2k}{2^k}n+\frac{4(k+1)}{2^k}n-4n\\
											& = (r-2) n + n\left(\frac{2k+2-r}{2^k}\right)\\
											& \geq (r-2) n,
\end{align*}
where last inequality holds because $r\leq 2k+1$.

Since $(T, \alpha)$ is stable, player $v$ has no incentive to buy the edge $(v,c_1)$. Hence, $\alpha - (r-2)n \geq 0$, i.e., $r\leq {\alpha}/{n}+2$. This implies that the length of the path from the centroid to any leaf of $T$ is at most ${\alpha}/{n}+4$. Thus, the diameter of $T$ is less than or equal to $2{\alpha}/{n}+8$. Since in every tree of $n$ vertices, there are $2(n-1)$ distinct pair of vertices at distance 1, while the other $(n-1)(n-2)$ pairs are at distance of at most $2{\alpha}/{n}+8$, the upper bound on the social cost of $T$ is \vspace*{-0.1cm}
\begin{align*}
\completecost(T) 	& = \alpha(n-1)+\left(\frac{2\alpha}{n}+8\right)(n-1)(n-2)+2(n-1) \\
					& = \frac{8n^2+(3\alpha-14)n-4\alpha}{n}\cdot (n-1).
\end{align*}\vspace*{-0.1cm}
The cost of the social optimum on $n$ nodes, which for $\alpha\geq 2$ is the star~\cite{Fab03}, is
\begin{equation*}
\completecost(S_n)=\alpha(n-1)+2(n-1)(n-2)+2(n-1)=\big(2n+\alpha-2\big)(n-1).
\end{equation*}
Therefore,
\begin{equation*}
\frac{\completecost(T)}{\completecost(S_n)} = \frac{8n^2+(3\alpha-14)n-4\alpha}{2n^2+(\alpha-2)n} = 3+\frac{2n^2-8n-4\alpha}{2n^2+(\alpha-2)n}. \qedhere
\end{equation*}
\end{proof}

\section{Conclusion}
In this paper we have opened a new line of attack on settling the tree conjecture and on proving a constant price of anarchy for the network creation game for all~$\alpha$. Our technique is orthogonal to the known approaches using bounds on the average degree of vertices in a biconnected component. We are confident that our methods can be refined and/or combined with the average degree technique to obtain even better bounds -- ideally proving or disproving the conjectures. 

Another interesting approach is to modify our techniques to cope with the so-called max-version of the network creation game~\cite{De07}, where agents try to minimize their maximum distance to all other nodes, instead of minimizing the sum of distances. Also for the max-version it is still open for which $\alpha$ all stable networks must be trees. 

~

\emph{Acknowledgement: We thank our anonymous reviewers for their helpful comments.}


	\bibliographystyle{plainurl}
	\bibliography{ncg_arxiv_new}

\end{document}